\documentclass[11pt, a4paper]{article}
\usepackage[utf8]{inputenc}

\usepackage[margin = 1in]{geometry}
\usepackage{amsfonts,amsmath,amssymb,amsthm}
\usepackage{graphicx,subcaption}
\usepackage{booktabs}
\usepackage{lipsum}
\usepackage{hyperref}

\usepackage{multicol}
\hypersetup{
   colorlinks = true,
   citecolor = red,
   urlcolor = blue,
   linkcolor = blue
}

\usepackage{mathrsfs}
\usepackage{multirow}
\usepackage[makeroom]{cancel}
\usepackage{tikz}
\usepackage{framed,color}
\definecolor{shadecolor}{rgb}{1,0.8,0.3}
\usepackage{fancybox}
\usepackage{caption}
\usepackage{algorithm,algorithmic} 
\usepackage{multicol}
\usepackage{aliascnt}
\usepackage{setspace}
\usepackage{longtable}
\usepackage{cite}

\usetikzlibrary{decorations.markings,arrows}

\allowdisplaybreaks

\title{Demand Analysis with a Thin Price Sample}

\date{}

\author{Monitirtha Dey\thanks{Email: monitirthadey3@gmail.com}, Arpan Kumar\thanks{Email: arpank05@gmail.com} \& Diganta Mukherjee\thanks{Email: digantam@hotmail.com} \thanks{We are grateful to Prabir Chaudhuri for his help on the conceptualisation and execution of the work. The usual caveat applies.} \\
        Indian Statistical Institute, Kolkata}

\begin{document}
\maketitle

\newtheorem{lem}{Lemma}
\newtheorem{cor}{Corollary}
\newtheorem{thm}{Theorem}

\newtheorem{theorem}{Theorem}

\newaliascnt{lemmaa}{theorem}
\newtheorem{lemmaa}[lemmaa]{Theorem}
\aliascntresetthe{lemmaa}
\providecommand*{\lemmaaautorefname}{Theorem}

\begin{abstract}
For about 125 items of food, the Consumer Expenditure Survey (CES) schedule of the Indian National Sample Survey asks the interviewer to obtain both quantity and value of household consumption during the reference period from the respondent. This would appear to put a great burden on the respondent. But it is likely that the price usually paid is almost the same within each first stage unit (fsu). The present work proposes a new sampling scheme to estimate demand elasticities of essential food items. While the conventional sampling method used in practice (e.g. in NSS consumer expenditure survey) involves seeking price information from many households sampled from a fsu, the proposed procedure involves only one household chosen randomly from every fsu for price data collection and thus requires much less interview burden. Using unit records for vegetable items in the NSS’s 2011-12 CES, our results show that in spite of requiring much less data, the new scheme captures the household food consumption behavior as precisely as before.
\end{abstract}
\thispagestyle{empty}

\section{Introduction}
India is the world's second most populous country. An analysis of household food consumption pattern and its response to changes in income and prices in one of the largest producing and consuming countries is essential to estimate the future demand of agricultural products in the country and also particularly important for agricultural products exporting countries (\cite{Kumar}, \cite{Agbola}).

An important Household consumption expenditure survey is the Consumer Expenditure Survey conducted by the Indian National Sample Survey Office (NSSO). The NSS consumer expenditure schedule obtains, through face-to-face interview of a large random sample of households spread over rural and urban areas of every district in India, information on more than 130 food items. For most food items, both quantity and value of consumption during the reference period are recorded.  In the last two surveys for which results are available, the reference period, for a large number of food items including all vegetables, was “last week” for one half of the sample households and “last month” for the other half.
In recent years there has been growing concern over the time needed to canvass the NSS CES and the time that the average household in India is willing to provide for the interview. While research continues on the feasibility of shortening the schedule in terms of number of items of consumption listed (without loss to the survey objectives), surprisingly little attention has been focussed on the double effort of obtaining both quantity and value of consumption, which in the last survey was done for 125 of the elementary items of food listed, 14 items of tobacco, intoxicants, etc., 12 items of fuel, and 30 items of clothing, bedding and footwear. However, abridged or partially abridged schedules have been tried out experimentally. Abridging the schedule usually takes the form of compressing groups of items into single items. 

In this work, we investigate that if we want to study the consumer expenditure pattern of a homogeneous society and ask one random household about the price of goods available there (thus using a much thinner sample of price data) and assume that this price holds for the whole society, then is it enough to predict the demand pattern of the society? Or do we have to collect this information from each of the households available there?

We propose a new sampling scheme to estimate demand elasticities of essential food items in this paper. While the conventional sampling method used in practice (e.g. in NSS consumer expenditure survey) involves seeking price information from all households sampled from a fsu, our proposed procedure involves only one household chosen randomly from every fsu for price data collection and thus requires much less interview burden. Using unit records for vegetable items in the NSS’s 2011-12 CES, we aim to test whether the new scheme captures the household food consumption behavior as precisely as before.

In the next section we describe the data used for analysis. Some preliminary analysis of data and relevant theoretical results are developed across several subsections in Section 3. The main analysis is presented in section 4 along with some concluding observations.

\section{The Data}
We use the data from household consumer expenditure survey 68th round \cite{NSSO} to analyse the distribution of the demand for several food items. Conducted by National Sample Survey Organisation(NSSO) this survey covered, during 2011-12, 101651 households in 7469 villages and 5268 urban blocks spread over the entire country.

The household consumer expenditure survey provides estimates of average household monthly per capita consumer expenditure (MPCE henceforth), the distribution of households and persons over the MPCE range, and the break-up of average MPCE by commodity group, separately for the rural and urban sectors of the country, across different socio-economic groups and across various states and union territories \cite{NSSO}. The indicators of monthly per capita consumption spending play a pivotal part in assessing standard of living,  shifting priorities in terms of baskets of goods and services across different strata of the population by providing the budget shares of different commodity groups for the rural and urban population separately.

Moreover, the estimated budget shares of a commodity at different MPCE levels make possible the study of consumption elasticity or responsiveness of demand for the commodity to change in purchasing power. A better understanding of demand elasticities facilitates prediction of future demand of food items under different scenarios of expenditure (considered as a proxy for income) and prices and may be useful to the policy makers on policy formulations \cite{Kumar}.

The NSS consumer expenditure schedule is organized in blocks in which consumption of food and non-food items, including services, is recorded.
For all items listed in the schedule, consumption expenditure is recorded. For some items, including nearly all food items, quantity consumed is also recorded. This study was confined to those food items for which the NSS consumer expenditure survey collects information on both quantity and value of consumption. Since the NSS schedule of enquiry uses around 150 items for food, individual food items have very small shares in the total consumer expenditure. It was decided to limit the analysis to those food items which had a share in household consumer expenditure (estimated by the survey) of at least 0.1\% at all-India level, separately for rural India and urban India. This gave a set of 29 food items as in Table \href{tab:my_label}{1}.
 
 \begin{table}[H]
     \centering
 \begin{tabular}{|c|l||c|l|}
 \hline   \textbf{Item No.}  & \textbf{Item Name} & \textbf{Item no.} & \textbf{Item Name}  \\
   \hline \hline  {\bf 101} & Rice (PDS)  & {\bf 181} &  Mustard Oil\\
    \hline {\bf 102}  & Rice (Other sources) & {\bf 184} & Refined Oil \\
     \hline {\bf 108}   & Wheat/Atta (Other sources) & 191 & Fish, Prawn\\
\hline    {\bf 111}     & Suji, Rawa & {\bf 192} & Goat Meat/Mutton \\
      \hline {\bf 115}   &  Jowar \& it's products & {\bf 195} & Chicken \\
    \hline      {\bf 140} & Arhar, Tur & {\bf 200} & Potato\\
          \hline 141 & Gram:Split & {\bf 201} & Onion\\
        \hline    {\bf 143} & Moong & {\bf 202} &  Tomato\\
        \hline   {\bf 144}  & Masur & 217 &  Other Vegetables\\
        \hline   {\bf 145}  & Urd & {\bf 220} & Banana\\
        \hline    {\bf 151} & Besan & 270 & Tea: Cups\\
        \hline {\bf 160}    & Milk: Liquid (litre) & {\bf 271} & Tea Leaf (gm) \\
        \hline    {\bf 164} & Ghee & 280& Cooked meals purchased\\
        \hline   {\bf 170}  & Salt & 292 &  Papad, Bhujia, Namkeen\\
        \hline   {\bf 172}  & Sugar & & \\
        \hline
 \end{tabular}
  \caption{List of the 29 food items chosen for our study}
     \label{tab:my_label}
 \end{table}

For each of these items we are given the data about several factors like:
\begin{itemize}
    \item \textbf{Sector} i.e, that is the data is collected from which kind of region which is indicated by (i) rural area and (ii) urban area.
    \item \textbf{State} i.e, which state the data are collected from. Here each state have been associated with a two digit code.
    \item \textbf{Household Id} i.e, identification of the sample household.
    \item \textbf{Size of the household} i.e, 
    the total number of persons in the household.
    \item \textbf{Value of the Item} Actually, the value of an item is what someone is freely willing to give for it and someone is willing to take for it. So here to measure the valuation we have the total monthly expenditure for a particular item for a particular household. 
    \item \textbf{Quantity of the Item} i.e, the total amount of the item consumed by the household each month.
    
    Once we have the value and quantity of an item then we can easily obtain the price by dividing value by quantity.
    \item \textbf{Monthly per Capita Expenditure} Normally, the concept of per capita income – or per capita (overall) expenditure, if income data are not available – is the measure used for comparison of average living standards between countries, between regions, and between social or occupational groups. For studies of poverty and inequality within populations, however, average income or average expenditure is not sufficient. One needs to assign a value that indicates level of living to each individual, or at least to each household, in a population in order to know the level of inequality in living standards of the population, or the proportion living in poverty.
    \\The NSS concept of MPCE, therefore, is defined first at the household level (household monthly consumer expenditure $\div$ household size). This measure serves as the indicator of the household’s level of living.
    \\Next, each individual’s MPCE is defined as the MPCE of the household to which the person (man, woman or child) belongs. This assigns to each person a number representing his or her level of living. The distribution of persons by their MPCE (i.e., their household MPCE) can then be built up, giving a picture of the population classified by economic level.
\end{itemize}

\section{Analysis}
\subsection{Analyzing the data}
“Tea: cups” (i.e, prepared tea) (Item $270$) was excluded from our analysis because in NSS consumer expenditure schedule quantity of tea was recorded in number of cups, which meant that the unit of quantity was highly variable. “Cooked meals purchased” (Item $280$) (food purchased as meals, with quantity recorded in number of meals) was excluded for the same reason.

It is evident from the list of 29 items mentioned before that not each of them represents a single food item. For example, the item “fish, prawn” (Item $191$) consists of prawn and several other types of fishes. Hence the price of this item depends on prices of all types of fishes and therefore has a very broad variation in price. This broad variation is also evident from the histograms of ratio of minimum and maximum prices (price-ratio henceforth) of the item reported in a fsu. One observes from Figure 1 that there is a significant mass in the left-side of the histogram, indicating that there is a broad price variation.

\begin{figure}[H]
  \centering
 \begin{subfigure}{\textwidth}
  \begin{minipage}[b]{0.4\textwidth}
    \includegraphics[width=\textwidth]{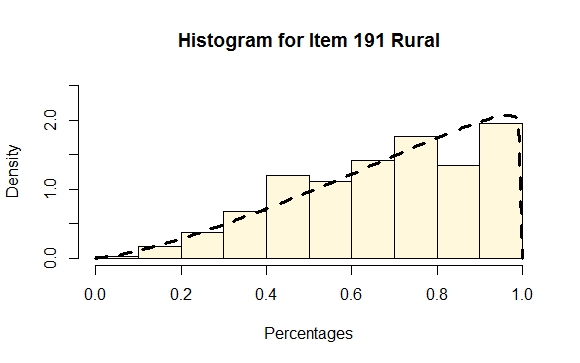}
  \end{minipage}
  \hfill
  \begin{minipage}[b]{0.4\textwidth}
    \includegraphics[width=\textwidth]{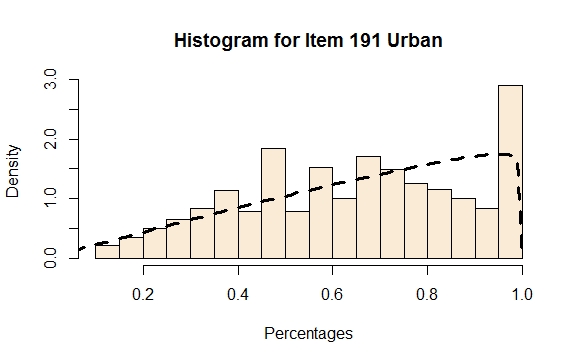}
  \end{minipage}
  \caption{Distribution of the price-ratio for the item Fish, Prawn}
  \end{subfigure}
  \begin{subfigure}{\textwidth}
  \begin{minipage}[b]{0.4\textwidth}
    \includegraphics[width=\textwidth]{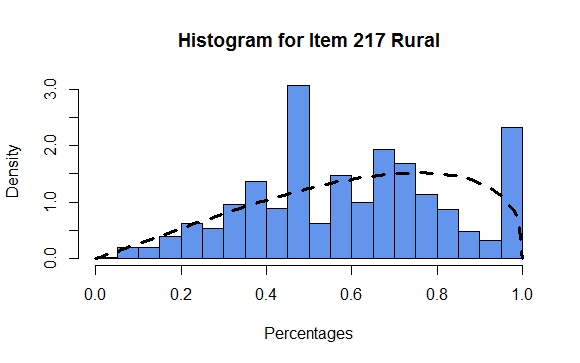} 
  \end{minipage}
  \hfill
  \begin{minipage}[b]{0.4\textwidth}
    \includegraphics[width=\textwidth]{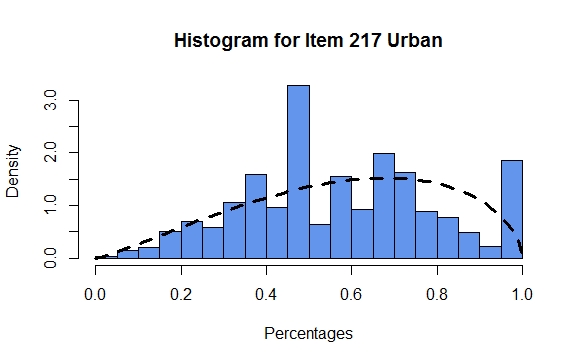}
  \end{minipage}
  \caption{Distribution of the price-ratio for the item Other Vegetables}
  \end{subfigure}
  \begin{subfigure}{\textwidth}
  \begin{minipage}[b]{0.4\textwidth}
    \includegraphics[width=\textwidth]{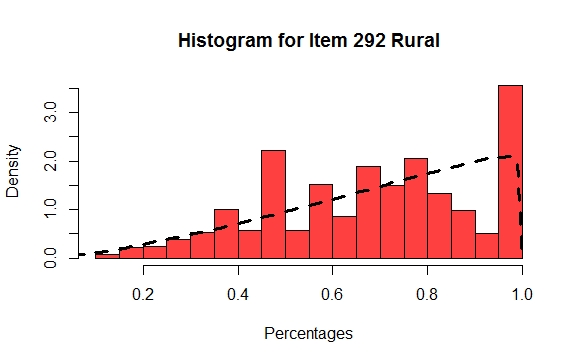}
  \end{minipage}
  \hfill
  \begin{minipage}[b]{0.4\textwidth}
    \includegraphics[width=\textwidth]{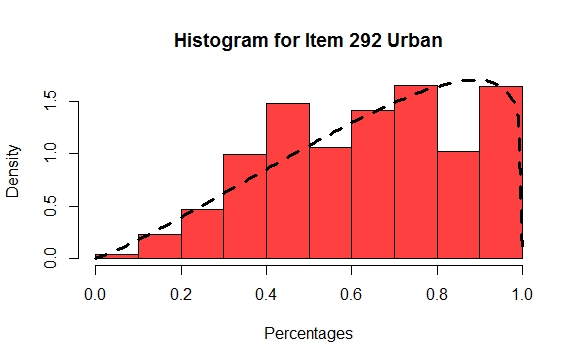}
  \end{minipage}
  \caption{Distribution of the price-ratio for the item Papad, Bhujia, Namkeen}
  \end{subfigure}
  \begin{subfigure}{\textwidth}
  \begin{minipage}[b]{0.4\textwidth}
    \includegraphics[width=\textwidth]{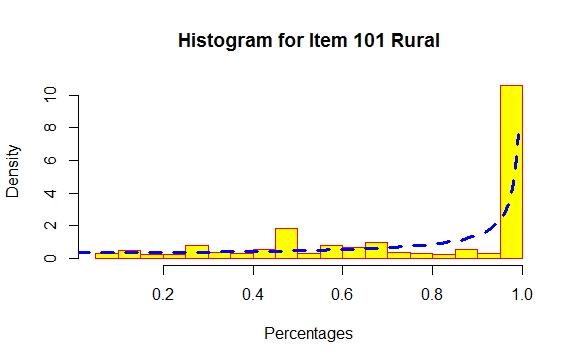} 
  \end{minipage}
  \hfill
  \begin{minipage}[b]{0.4\textwidth}
    \includegraphics[width=\textwidth]{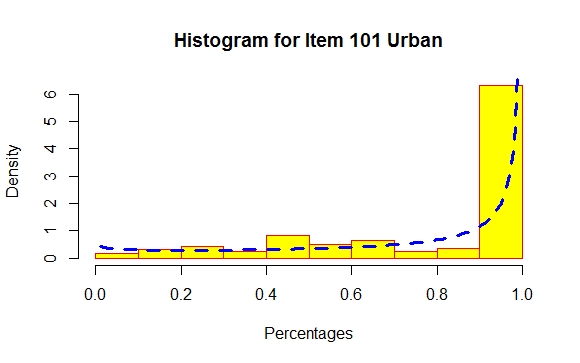}
  \end{minipage}
  \caption{Distribution of the price-ratio for the item Rice (PDS)}
  \end{subfigure}
  \begin{subfigure}{\textwidth}
  \begin{minipage}[b]{0.4\textwidth}
    \includegraphics[width=\textwidth]{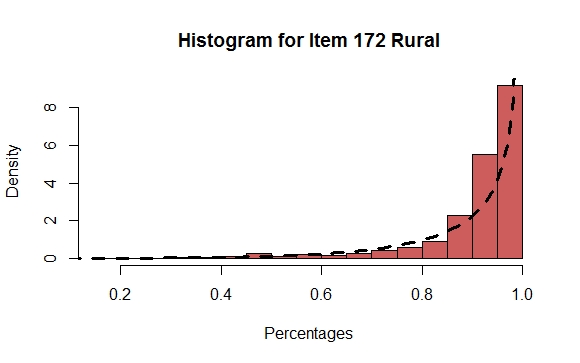}
  \end{minipage}
  \hfill
  \begin{minipage}[b]{0.4\textwidth}
    \includegraphics[width=\textwidth]{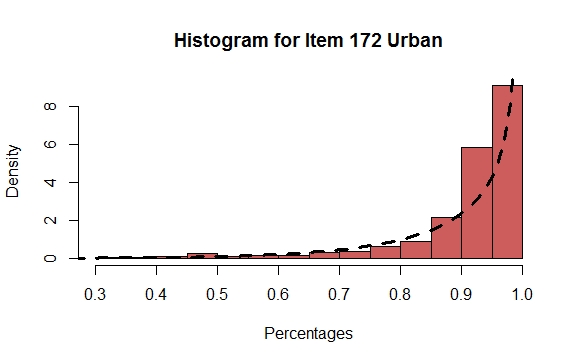}
  \end{minipage}
  \caption{Distribution of the price-ratio for the item Sugar}
  \end{subfigure}
  \caption{Distributions of price-ratio for illustrative items}
  \label{fig:1}
\end{figure}

The items “other vegetables” (Item $217$) and “papad, bhujia, namkeen” (Item $292$) also represented a heterogeneous collection of items which would clearly have a very wide price variation. Hence, we have excluded the three items “fish, prawn”, “other vegetables” and “papad, bhujia, namkeen” from our analysis. The histograms for the price-ratios of other items are markedly concentrated in a region close to $1$. The price-ratio histograms for the items “Rice(PDS)” (Item $101$) and “Sugar” (Item $172$) also are given in Figure \href{fig:1}{1}.

Finally, “gram, split” was excluded as there were 4 other pulses (items with numbers 140, 143, 144, 145) in the list, and a fifth did not seem necessary, as the four pulses arhar, moong, masur and urd in 2011-12 together made up about 67\% of consumption of pulses and pulse products in urban India and 63\% in rural India \cite{NSSO}. 
\vspace{2mm}
\\ Our interest is now to check whether the proposed sampling scheme can capture the same distribution of household consumption behaviour as the old scheme for the remaining 23 items i.e, item no. 101, 102, 108, 111, 115, 140, 143, 144, 145, 151, 160, 164, 170, 172, 181, 184, 192, 195, 200, 201, 202, 220, 271 respectively (highlighted in Table 1).  

\subsection{Sampling Scheme}

The proposed sampling scheme involves picking one household randomly from every fsu and assuming that the price responses of the other households in that particular fsu are the same as that of the selected household. At the same time it should be ensured that the food items under consideration are represented in a large proportion in our sample. 

The total number of fsu’s in our data is $N=12734$. For a particular item, let $Y_{i}$ represent the occurrence of that item in a sample household selected from the $i$’th fsu for $1 \leqslant i \leqslant N$. So, $Y_{i}$ follows $\operatorname{Bernoulli}({p_i})$ independently of each other where $p_{i}$ is the probability of the occurrence of the item. Let $X=\sum_{i=1}^{N} Y_{i}$. We are interested in $\mathbb{P}(X \geqslant Nq)$ (\textit{prevalence probability} henceforth) for a suitable choice of $q \in (0,1)$, which indicates the probability that the item is consumed in at least one household in $100q \%$ of the $N$ fsu’s. 

As $X$ is a sum of independent Bernoulli random variables with possibly different success probabilities, it follows that $X$ follows a Poisson Binomial distribution \cite{ChenAoP}. Since working with the exact p.m.f of a Poisson Binomial distribution is mathematically inconvenient and $N$ is very large, we approximate the desired probability with the help of the following theorem.

\begin{theorem}\label{ref1}
Let $N \in \mathbb{N}$ and for $1 \leq i \leq N$, $Y_{i} \sim  \operatorname{Bernoulli}({p_i})$ independent of each other. Suppose $X = \sum_{i=1}^{N}Y_i$ and $S_N = \sum_{i=1}^{N}p_i(1-p_i)$. Then, 
\begin{equation}
    \dfrac{1}{\sqrt{S_n}}\left(X-\sum_{i=1}^{N}p_{i}\right) \xrightarrow{d} \mathcal{N}(0,1) \hspace{2mm} \text{as} \hspace{2mm} N \to \infty \label{eq1}
\end{equation}
provided $0 < \liminf{p_{i}} \leqslant \limsup {p_{i}} < 1$.
\end{theorem}

\begin{proof}
We have,
$$\mathbb{E}(X) = \sum_{i=1}^{N}\mathbb{E}(Y_{i})= \sum_{i=1}^{N} p_i  \quad \text{ and } \operatorname{Var}(X) = \sum_{i=1}^{N}\operatorname{Var}(Y_{i}) = \sum_{i=1}^{N}p_i(1-p_i) =S_N.$$. Also,
\begin{equation*}
    \left|Y_{i}-p_i\right| = \begin{cases}1-p_i \hspace{2mm}&\text{with probability} \hspace{2mm} p_i\\
    p_i \hspace{2mm}&\text{with probability} \hspace{2mm} 1-p_i.\\
    \end{cases}
\end{equation*}
This gives
\begin{align*}
\mathbb{E}\left[\left|Y_{i}-p_i\right|^{3}\right] & = p_i(1-p_i)^{3}+(1-p_i)p_i^{3} \\
& = p_i(1-p_i)\left[(1-p_i)^{2} + p_{i}^{2}\right]\\
& \leqslant p_i(1-p_i) \left[(1-p_i) + p_{i}\right]\\
& = p_i(1-p_i).
\end{align*}

\noindent Summing over $i$ we obtain,
\begin{align*}
\sum_{i=1}^{N} \mathbb{E}\left[\left|Y_{i}-p_i\right|^{3}\right] \leqslant \sum_{i=1}^{N} p_i(1-p_i).
\end{align*}
\noindent Therefore, 
$$\displaystyle \dfrac{\displaystyle \sum_{i=1}^{N} \mathbb{E}\left[\left|Y_{i}-p_i\right|^{3}\right]}{S_{N}^{3/2}} \leqslant  \displaystyle \dfrac{1}{\left[\displaystyle \sum_{i=1}^{N} p_i(1-p_i)\right]^{1/2}}$$
Hence, when $0 < \liminf{p_{i}} \leq \limsup {p_{i}} < 1$, using Lyapunov Central Limit Theorem we obtain \eqref{eq1}. \end{proof}
Hence, for sufficiently large $N$, the suitably standardized version of $X$ approximately follows the standard normal distribution. So \eqref{eq1} implies, for any specific $q \in (0,1)$,
\begin{equation}
    \mathbb{P}\left(\frac{X}{N} > q\right) \approx 1 - \Phi\left(\dfrac{Nq-\sum_{i=1}^{N}p_{i}}{\sqrt{S_n}}\right), \label{eq2}
\end{equation}
where $\Phi$ denotes the c.d.f. of the standard normal distribution. Thus, we can approximate the prevalence probability of the item by the quantity in r.h.s of \eqref{eq2}. Table \ref{tab:Tab2} depicts the values of the probabilities for different items for different choices of $q$, viz. $0.5,0.4,0.3$. It turns out that for $q=0.5$, 11 items have zero prevalence probabilities of being consumed in at least $100 q \%$ of the $N$ fsu's whereas for $q=0.3$ only 4 items have zero prevalence probabilities.

\begin{table}[ht]
\centering
\begin{tabular}{|c||c||c|c|c|}
\hline
\multirow{2}{*}{Item No.} & \multirow{2}{*}{Item Name} & \multicolumn{3}{c|}{Probability in r.h.s of \eqref{eq2}} \\ \cline{3-5} 
 &  & $q=0.5$ & $q=0.4$ & $q=0.3$ \\ \hline \hline
101 & Rice (PDS) & 0 & 0 &  \\ \hline
111 & Suji, Rawa & 0 & 0 &  \\ \hline
115 & Jowar \& it's products & 0 & 0 & 0 \\ \hline
144 & Masur & 0 & &  \\ \hline
145 & Urd & 0 & &  \\ \hline
151 & Besan & 0 & 0 &  \\ \hline
164 & Ghee & 0 & 0 & 0 \\ \hline
184 & Refined Oil & 0 & &  \\ \hline
192 & Goat Meat/ Mutton & 0 & 0 &  0\\ \hline
195 & Chicken & 0 & 0 &  0\\ \hline
220 & Banana & 0 & &  \\ \hline
\end{tabular}
\caption{List of items with zero prevalence probabilities for different choices of $q$ }
\label{tab:Tab2}
\end{table}

\subsection{The Demand Model}

Several demand models are available to estimate the price and income elasticities of demand for various food commodities. Linear Expenditure System (by Stone \cite{Stone}), and Almost Ideal Demand System (AIDS) (by Deaton and Muellbauer \cite{Deaton}) are two widely studied demand models. The interested reader may see the papers by Agbola \cite{Agbola}, Pollak \& Wales \cite{Pollak}, Parks \cite{Parks}, Ham \cite{Ham},  Chalfant \cite{Chalfant} and Blanciforti \& Green \cite{Blanciforti} for more works on these two demand models.
We consider the model upon the actual price along with other factors first:
\begin{equation}
    \text{log} Q = \alpha_i + \beta_j + \gamma_1 S +\gamma_2 \text{log} P +\gamma_3 \text{log} E \label{eq3}
\end{equation}
where $\alpha$ is a factor variable denoting the sector, $\beta$ is a factor variable denoting the state, $S$ is the size of the household, $E$ is the monthly per capita expenditure of a household, $P$ is the price and Q is demand of the item.

Then we fit the same model considering the randomly chosen price $P^*$ in place of actual price $P$ i.e,
\begin{align}
    \text{log} Q & = \alpha_i + \beta_j + \delta_1 S +\delta_2 \text{log} P^* +\delta_3 \text{log} E \notag\\
    & = \alpha_i + \beta_j + \delta_1 S  +\delta_3 \text{log} E +\delta_4 \text{log} P + \delta_5 \text{log}\dfrac{P^*}{P} \label{eq4}
\end{align}
 Here log $\dfrac{P^*}{P}$ is the elasticity.
 
 Households belonging to various fractile classes of the MPCE distribution represent a series of sub-populations with gradually increasing level of living. The variation in the budget share of any particular food item across MPCE fractile classes thus enables the study of variation in consumption behaviour with rise in level of living. For any particular item of consumption, the share in the household budget is the same as the ratio of per capita
 expenditure on the item to MPCE.
 
 To achieve our goal we have to check if the predicted distribution of the share depending upon the response variable in both the models are same or not. Another way is to check is whether $\gamma_5$, i.e. the coefficient to the elasticity term, is zero or not.

\subsection{Bias correction}
Standard regression models assume that the predictor variables involved in the model are measured exactly, or observed without any error. In contrast, we have here a measurement error model (\cite{Stefanski}, \cite{Reeves}, \cite{Chen}), i.e, we have independent variables measured with error, since we are randomly choosing one value from each fsu. To account for these errors, it is expected that a bias would be introduced in the regression coefficients.

A measurement error in a predictor variable is called “classical” if it is independent of the latent (unobserved) true variable; otherwise it is called “non-classical”. Various methods are used to treat classical and non-classical measurement errors. In this article, we shall work with classical measurement errors.

We consider the following linear regression model:
$$Y =X^{*}\mathbf{\beta} + \mathbf{\epsilon}$$
where instead of observing the true predictor $X^{*}$, we observe $X$ given by
$$X =  X^{*}  + V$$ 
where $V$ is a matrix of measurement errors. Here we assume the following:
\begin{itemize}
    \item[(A1)] The true values $X^{*}$ and the model errors $\epsilon$ are independent.
    \item[(A2)] The true values $X^{*}$ and the measurement errors $V$ are independent, i.e the measurement errors are classical.
    \item[(A3)] The measurement errors $V$ and the model errors $\epsilon$ are independent.
    \item[(A4)] The columns of $V$ are independent, i.e, the measurement error in one predictor is independent of measurement error in another predictor.
\end{itemize}
\begin{theorem}
Under the assumptions (A1)-(A4), 
$$\hat{\beta}^{*} = \big[I - (X'X)^{-1} V'V\big]^{-1}(X'X)^{-1}X'Y$$
is a consistent estimator of $\beta$.
\end{theorem}
\begin{proof}
We at first consider the usual estimator of $\beta$:
\begin{align*}
    \hat{\beta} & = (X'X)^{-1}X'Y\\
    & = (X'X)^{-1}X'(X^{*}\beta + \epsilon)\\
    & = (X'X)^{-1}X'X^{*}\beta + (X'X)^{-1}X'\epsilon\\
    & = (X'X)^{-1}X'(X-V)\beta + (X'X)^{-1}(X^{*}+V)'\epsilon\\
    & = \beta - (X'X)^{-1}X'V\beta + (X'X)^{-1}X^{*'}\epsilon + (X'X)^{-1}V'\epsilon\\
    & = \beta - (X'X)^{-1}(X^{*}+V)'V\beta + (X'X)^{-1}X^{*'}\epsilon + (X'X)^{-1}V'\epsilon\\
    & = \beta - (X'X)^{-1}V'V\beta - (X'X)^{-1}X^{*'}V\beta + (X'X)^{-1}X^{*'}\epsilon + (X'X)^{-1}V'\epsilon\\
    & = \big[I - (X'X)^{-1}V'V\big]\beta - (X'X)^{-1}X^{*'}V\beta + (X'X)^{-1}X^{*'}\epsilon + (X'X)^{-1}V'\epsilon\\
    & \stackrel{P}{\to} \big[I - (X'X)^{-1} V'V\big]\beta
\end{align*}
where the last step above use the facts that $X^{*'}V \stackrel{P}{\to} 0$, $X^{*'}\epsilon \stackrel{P}{\to} 0$ and $V'\epsilon \stackrel{P}{\to} 0$ which hold because of our first three assumptions mentioned earlier. Hence,  $\mathbf{\hat{\beta}}$ \textit{is an inconsistent estimator of} $\mathbf{\beta}$. A consistent estimator based on $\hat{\beta}$ may be obtained as
$$\hat{\beta}^{*} = \big[I - (X'X)^{-1} V'V\big]^{-1}\hat{\beta},$$
completing the proof.
\end{proof}
In our case, measurement error is present only in the predictor indicating price. So, here each column of $V$ (except the second column) is a zero column. Consequently, $V'V$ has a non-zero entry in second diagonal position and zero everywhere else. Now,
\begin{align*}
    \hat{\beta}^{*} &= \big[I - (X'X)^{-1} V'V \big]^{-1}\hat{\beta}\\
    &= \big[X'X- V'V\big]^{-1}X'Y\\
    &= \big[X'X- V'V\big]^{-1}X'X\hat{\beta}\\
    &= \hat{\beta} + \big[X'X- V'V\big]^{-1}\cdot V'V \hat{\beta}\\
    &= \hat{\beta} + \big[X'X- V'V\big]^{-1}\cdot n \cdot (0 \hspace{1.9mm} \hat{\beta_1} \hspace{1.9mm}\cdots \hspace{1.9mm} 0)'
\end{align*}
This $\hat{\beta}^{*}$ is a consistent estimator of $\beta$, as mentioned earlier.

\section{Results and Observations}

Suppose $F$ and $G$ are the distribution functions of share based on the actual price and the randomly chosen price respectively. We wish to see if these two distributions are same, i.e we want to test
$$H_0 : F=G$$
using the usual \textit{Kolmogorov-Smirnof test}. 

As mentioned in Section 3.2, we are creating the randomized response variable by just picking one household randomly from every fsu. This algorithm may result in a biased acceptance or rejection of the test i.e, suppose for a fsu one of the extreme values happens to be chosen, then it will it will certainly affect the outcome of the test. So to avoid this bias, we have repeated this test 1000 times with each test considering newly chosen household sample with replacement so that all households are adequately represented.

\begin{theorem}
Let $p_{(1)} \leqslant p_{(2)} \leqslant  \ldots \leqslant p_{(1000)}$ be the ordered p-values obtained from the 1000 repetitions. The rejection criterion
\begin{equation}
    \text{reject $H_{0}$ if} \hspace{2mm}p_{(62)} < 0.05 \hspace{2mm} \text{and accept otherwise} \label{eq6}
\end{equation}
controls Type 1 error at 0.05 level.
\end{theorem}
\begin{proof}
$H_0$ is rejected at level 0.05 if we get a certain number (say $c$) of rejections among the $1000$ tests. We have to find the value of $c$ such that the size of our test remains at $0.05$. For $1 \leqslant i \leqslant N$, consider the variable 
$$Z_i=\mathbb{I}_{\{p_i<0.05\}}$$
where $p_i$ is the p-value of the test performed with the $i^{th}$ sample.
Note that under $H_0$, $p_i$ follows $\text{Uniform}(0,1)$ for $i=1, \ldots, 1000$. Hence, under $H_{0}$, 
$$Z_i\thicksim \text{Bernoulli}(0.05)\ \forall i=1(1)1000.$$
The 1000 samples are drawn from the full data with replacement and are therefore independent given the full data. Here by \textbf{full data} we mean the complete data available i.e, the data accumulated by NSS from all the possible households for each of the items. This implies that the corresponding p-values $(p_i)$ are independent. Hence $Z_i$'s are independent. Consequently,
$$Z=\sum_{i=1}^{1000} Z_i \thicksim \text{Binomial}(1000,0.05).$$
Our objective is to find a $c \in \mathbb{N}$ such that $\mathbb{P}_{H_{0}}(p_{(c)} \leqslant 0.05) \leqslant 0.05$. One also observes that
$$\{Z>c\}\iff \{p_{(c)}<0.05\}.$$ 
Hence it is enough to find a $c$ such that 
\begin{equation}
\mathbb{P}_{H_0}(Z>c) \leqslant 0.05 \label{eq7}
\end{equation}
Towards this, we claim that if we find a $c$ such that, $\mathbb{P}_{H_0}(Z>c|\text{full data})\leqslant 0.05$ then the same value of $c$ will satisfy \eqref{eq7}. This is justified in the following chain of equations.
\begin{align*}\mathbb{P}_{H_0}(Z>c) & = \mathbb{E}_{H_0}(\mathbb{I}_{\{Z>c\}}) \\
& = \mathbb{E}_{\text{full data}}\left[\mathbb{E}_{H_0}(\mathbb{I}_{\{Z>c\}}|\text{full data})\right] \\
& = \mathbb{E}_{\text{full data}}\left[\mathbb{P}_{H_0}(Z>c|\text{full data})\right] 
\end{align*}
Therefore $\mathbb{P}_{H_0}(Z>c|\text{full data})\leqslant 0.05$ implies $\mathbb{P}_{H_0}(Z>c) \leqslant 0.05$. The minimum value of $c$ satisfying
$\mathbb{P}_{H_0}(Z>c|\text{full data})\leqslant 0.05$ is clearly the 95{th} quantile of the distribution of $Z$. By computation this is found to be $62$. In other words, $\mathbb{P}_{H_0}(Z>61) > 0.05$ but $\mathbb{P}_{H_0}(Z>62)\leqslant 0.05$. Equivalently, $\mathbb{P}_{H_0}(p_{(61)}<0.05) > 0.05$ and $\mathbb{P}_{H_0}(p_{(62)}<0.05)\leqslant 0.05$. 
Hence the test given by \eqref{eq6} controls Type 1 error at 0.05 level.
\end{proof}

We adopt the rejection criterion \eqref{eq6}. The results are presented in Table \href{tab:Tab3}{3}. It is seen that for cereals, pulses, edible oils, and items coming from milk-and-milk-products the distribution functions of share based on the actual price and the randomized price are statistically same. For four items (highlighted significance values), viz. onion (item 201), tomato (item 202), banana (item 220) and tea leaves (item 271), the distributions are significantly different. Thus, out of the 23 food items chosen for analysis, 19 items do not show any loss of information due to our thin price sample strategy. In fact, item 271 is a marginal case.

In Table \href{tab:Tab3}{3}, LCB($\delta_{5}$) and UCB($\delta_{5}$) denote the $25$'th and $975$'th ordered values of $\delta_{5}$ (given in \eqref{eq4}) obtained from 1000 repetitions, respectively. Then an empirical $95 \%$ confidence interval for $\delta_{5}$ can be obtained as [LCB($\delta_{5}$), UCB($\delta_{5}$)]. When the confidence interval for a particular item contains the value $0$, one expects the corresponding null to be true. We observe from Table \href{tab:Tab3}{3} that, these confidence intervals contain the value $0$ for 12 items and among them, only for banana (item 220) the null hypothesis is rejected. We also note that this confidence interval procedure empirically tests the hypothesis whether the mean of the 1000 $\delta_{5}$ values is zero for a particular item.

Similarly, an empirical $95 \%$ confidence interval for $\delta_{4}$ (given in \eqref{eq4}) is given as [LCB($\delta_{4}$), UCB($\delta_{4}$)] where LCB($\delta_{4}$) and UCB($\delta_{4}$) denote the $25$'th and $975$'th ordered values of $\delta_{4}$ obtained from 1000 repetitions. We observe that these confidence intervals contain the value of $\gamma_{2}$ (obtained from fitting the model \eqref{eq3}) for 12 items and among them, only for banana (item 220) the null hypothesis is rejected. In fact interestingly, it is the same set of 12 items mentioned above with respect to the discussion on $\delta_{5}$.

\begin{table}[H]
\centering
\begin{tabular}{|c||c||c||cc||ccc|}
\hline
\multirow{2}{*}{\begin{tabular}[c]{@{}c@{}}Item\\  No.\end{tabular}} & \multirow{2}{*}{\begin{tabular}[c]{@{}c@{}}Sample\\  Size\end{tabular}} & \multirow{2}{*}{\begin{tabular}[c]{@{}c@{}}$62^{nd}$ p-value\\ of K-S Test\end{tabular}} & \multicolumn{2}{c||}{Coefficient of elasticity} & \multicolumn{3}{c|}{Coefficient of price} \\ \cline{4-8} 
 &  &  & \multicolumn{1}{c|}{LCB$(\delta_5)$} & UCB$(\delta_5)$ & \multicolumn{1}{c|}{$\gamma_2$} & \multicolumn{1}{c|}{LCB$(\delta_4)$} & UCB$(\delta_4)$ \\ \hline
101 & 39417 & 0.992419 & \multicolumn{1}{c|}{-11.41} & -2.68 & \multicolumn{1}{c|}{-30.1} & \multicolumn{1}{c|}{-29.32} & -26.87 \\ \hline
102 & 75099 & 0.288596 & \multicolumn{1}{c|}{5.99} & 8.92 & \multicolumn{1}{c|}{-20.10} & \multicolumn{1}{c|}{-24.18} & -22.80 \\ \hline
108 & 68269 & 0.064759 & \multicolumn{1}{c|}{5.39} & 10.65 & \multicolumn{1}{c|}{-32.30} & \multicolumn{1}{c|}{-35.45} & -33.89 \\ \hline
111 & 34460 & 0.709766 & \multicolumn{1}{c|}{-2.16} & 1.93 & \multicolumn{1}{c|}{-15.83} & \multicolumn{1}{c|}{-16.28} & -15.36 \\ \hline
115 & 5733 & 1.000000 & \multicolumn{1}{c|}{-11.59} & -0.09 & \multicolumn{1}{c|}{-1.45} & \multicolumn{1}{c|}{-1.44} & 0.67 \\ \hline
140 & 56527 & 0.993570 & \multicolumn{1}{c|}{2.02} & 3.54 & \multicolumn{1}{c|}{-6.32} & \multicolumn{1}{c|}{-7.84} & -7.17 \\ \hline
143 & 49116 & 0.765490 & \multicolumn{1}{c|}{-0.89} & 0.91 & \multicolumn{1}{c|}{-6.52} & \multicolumn{1}{c|}{-6.82} & -6.24 \\ \hline
144 & 43670 & 0.998136 & \multicolumn{1}{c|}{-2.35} & -0.21 & \multicolumn{1}{c|}{-6.81} & \multicolumn{1}{c|}{-6.75} & -6.15 \\ \hline
145 & 38858 & 0.682071 & \multicolumn{1}{c|}{0.95} & 2.59 & \multicolumn{1}{c|}{-5.50} & \multicolumn{1}{c|}{-6.55} & -5.89 \\ \hline
151 & 36093 & 0.125716 & \multicolumn{1}{c|}{-0.93} & 1.51 & \multicolumn{1}{c|}{-11.88} & \multicolumn{1}{c|}{-12.22} & -11.67 \\ \hline
160 & 62095 & 0.787887 & \multicolumn{1}{c|}{-2.21} & 0.54 & \multicolumn{1}{c|}{-9.12} & \multicolumn{1}{c|}{-9.30} & -8.34 \\ \hline
164 & 17251 & 0.998243 & \multicolumn{1}{c|}{-0.38} & 0.28 & \multicolumn{1}{c|}{-0.70} & \multicolumn{1}{c|}{-0.74} & -0.63 \\ \hline
170 & 99062 & 0.104192 & \multicolumn{1}{c|}{-2.12} & 1.59 & \multicolumn{1}{c|}{-36.12} & \multicolumn{1}{c|}{-36.85} & -35.12 \\ \hline
172 & 84773 & 0.999291 & \multicolumn{1}{c|}{-0.79} & 4.05 & \multicolumn{1}{c|}{-15.1} & \multicolumn{1}{c|}{-16.77} & -14.79 \\ \hline
181 & 49912 & 1.000000 & \multicolumn{1}{c|}{-2.55} & 2.43 & \multicolumn{1}{c|}{-0.51} & \multicolumn{1}{c|}{-2.81} & -0.44 \\ \hline
184 & 42185 & 0.999854 & \multicolumn{1}{c|}{-0.81} & 0.40 & \multicolumn{1}{c|}{-2.72} & \multicolumn{1}{c|}{-2.87} & -2.43 \\ \hline
192 & 8574 & 0.869996 & \multicolumn{1}{c|}{-0.56} & 0.87 & \multicolumn{1}{c|}{-1.65} & \multicolumn{1}{c|}{-1.72} & -1.60 \\ \hline
195 & 23140 & 0.968405 & \multicolumn{1}{c|}{-0.45} & 0.69 & \multicolumn{1}{c|}{-3.26} & \multicolumn{1}{c|}{-3.36} & -3.20 \\ \hline
200 & 89169 & 0.962644 & \multicolumn{1}{c|}{-5.61} & -1.63 & \multicolumn{1}{c|}{-26.78} & \multicolumn{1}{c|}{-26.34} & -25.21 \\ \hline
201 & 94384 & \textbf{0.014069} & \multicolumn{1}{c|}{-5.09} & -2.42 & \multicolumn{1}{c|}{-22.04} & \multicolumn{1}{c|}{-21.32} & -20.48 \\ \hline
202 & 77711 & \textbf{0.000001} & \multicolumn{1}{c|}{-2.33} & -0.02 & \multicolumn{1}{c|}{-22.64} & \multicolumn{1}{c|}{-22.64} & -22.18 \\ \hline
220 & 42457 & \textbf{0.000000} & \multicolumn{1}{c|}{-0.02} & 0.01 & \multicolumn{1}{c|}{-0.15} & \multicolumn{1}{c|}{-0.15} & -0.15 \\ \hline
271 & 90472 & \textbf{0.047320} & \multicolumn{1}{c|}{0.05} & 0.30 & \multicolumn{1}{c|}{-1.53} & \multicolumn{1}{c|}{-1.72} & -1.55 \\ \hline
\end{tabular}
\label{tab:Tab3}
\caption{Item-wise significance results}
\end{table}

These results seem quite favourable for our thin sampling strategy. The departure of the precision of results from the usual situation turns out to be minimal and hence, given the substantial reduction in interview time (and hence cost), seems to lead to a more efficient sampling strategy.  As this is not restricted to a particular geographical context, it would be interesting to see whether our strategy works also for other countries' consumer expenditure surveys. We hope that our results will encourage analogous studies with other consumer expenditure data sets.

\bibliographystyle{plain}
\bibliography{Paper}

\end{document}